\documentclass[11pt]{article}

\usepackage{longtable}
\usepackage{textcomp}
\usepackage{amsmath}
\usepackage{amssymb}
\usepackage{mathrsfs}
\usepackage{amsfonts}
\usepackage{multirow}
\usepackage{cases}
\usepackage{graphicx}
\usepackage{tabularx}
\usepackage{color}
\usepackage{setspace}
\usepackage{bm}
\usepackage{indentfirst}

\newtheorem{Theorem}{Theorem}[section]
\newtheorem{Definition}{Definition}[section]
\newtheorem{Proposition}{Proposition}[section]

\newtheorem{Example}{Example}[section]

\newcommand{\alglist}{
\begin{list}{Step 1}
{\setlength{\leftmargin}{1.1 in}\setlength{\labelwidth}{1.0 in}} }

\newenvironment{proof}{{\textbf{Proof.}}\,}{\hfill$\hbox{\rule{5pt}{5pt}}$\\}

\begin{document}

\title{ An Irreducible Polynomial Functional Basis of Two-dimensional Eshelby Tensors
 }

\author{Zhenyu Ming \and  Liping Zhang\thanks{Department of Mathematical Sciences, Tsinghua University, Beijing 100084,  China. The first two authors were supported by the National Natural Science Foundation of China (Grant No. 11271221, 11771244). E-mail:  mingzy17@mails.tsinghua.edu.cn; lzhang@math.tsinghua.edu.cn}
\quad\quad Yannan Chen\thanks{School of Mathematics and Statistics, Zhengzhou University, Zhengzhou 450001, China (ynchen@zzu.edu.cn).
This author was supported by the National Natural Science Foundation of China (Grant No. 11571178, 11771405).}
}

\date{}
\maketitle

\begin{abstract}
Representation theorems for both isotropic and anisotropic functions are of prime importance in both theoretical and applied mechanics.
In this article, we discuss about two-dimensional Eshelby tensors (denoted as  $\bf M^{(2)}$), which have wide applications in many fields of mechanics. Based upon the complex variable method, we obtain an integrity basis of ten isotropic invariants of $\bf M^{(2)}$. Since an integrity basis is always a polynomial functional basis, we further confirm that this integrity basis is also an irreducible polynomial functional basis of $\bf M^{(2)}$.

{\bf Key\ words.}  Eshelby tensor, representation theorem, irreducible functional basis, isotropic invariant.

\end{abstract}


\section{Introduction}

The Eshelby problem for linear elasticity is to deal with the fields in an infinitely region $\Omega$ induced by releasing either transformation strains, or eigenstrains in a subdomain $\omega$, called an inclusion. Precisely addressing this objection, the strain field $\epsilon_{ij}(x)$ can be linearly expressed in the form
$$\epsilon_{ij}(x)=\Sigma^{\omega}_{ijkl}(x)\epsilon^0_{ij},$$
where $\epsilon^0_{ij}$ is a constant second order eigenstrain tensor, and the fourth order tensor $\Sigma^{\omega}_{ijkl}(x)$ is the Eshelby's tensor field, corresponding to a inclusion $\omega$. Eshelby \cite{eshelby1957,eshelby1959} proved that Eshelby's tensor field is uniform inside $\omega$ if the inclusion is elliptic and ellipsoidal in  two- or three-dimensional elasticity, respectively. In such a condition, the Eshelby's tensor field $\Sigma^{\omega}$ is called Eshelby tensor and this significant property is called Eshelby's uniformity. Eshelby tensors and Eshelby's uniformity have wide application in a great number of engineering and physical fields, such as elliptical and non-elliptical inclusions \cite{Jiang99,Michelitsch2003,Zou13}, matrix-inclusion composites \cite{huang1994,Pan04,Roatta1997,Wang18,Zou17}, non-uniform Gaussian and exponential eigenstrain within ellipsoids \cite{Sharma2003}, to name a few. Besides, Eshelby problem further became the subject of extensive studies \cite{Bacon1980,Buryachenko2001, Mura1987,Ting1996}.

 On the other hand, tensor function representation theory is well established and of prime importance in both theoretical and applied mechanic \cite{zheng1}. It was introduced to describe general consistent invariant forms of the nonlinear constitutive equations and to determine the number and the type of scalar variables
involved. In the latter half of the twentieth century, representations in complete and irreducible forms of vectors, second order symmetric tensors and second order skew-symmetric tensors for both isotropic and hemitropic invariants, were thoroughly investigated by Wang \cite{W-701,W-702,W-7071}, Smith \cite{Sm-71}, Boehler \cite{Boe-87}, and Zheng \cite{zheng1}. As for higher order tensors, Smith and Bao \cite{sbao}  presented minimal integrity bases for third and fourth order symmetric and traceless tensors. In 2014, an integrity basis with thirteen isotropic invariants of a symmetric third order three-dimensional tensor was presented by Olive and Auffray \cite{oa2014}.   It was noted that the Olive-Auffray integrity basis is actually a minimal integrity basis \cite{oa2017}.   In 2017, Olive, Kolev and Auffray \cite{olive2017} gave a minimal integrity basis of the elasticity tensors, with 297 invariants. Very recently, a number of new results appeared.  Liu, Ding, Qi and Zou \cite{ding17} gave a minimal integrity basis and irreducible functional basis of isotropic invariants of the Hall tensors.   Chen, Hu, Qi and Zou \cite{chen18} showed that any minimal integrity basis of a third order three-dimensional symmetric and traceless tensor is indeed an irreducible functional basis of that tensor. Chen, Liu, Qi, Zheng and Zou \cite{zchen18} presented an eleven invariant irreducible functional basis for a third order three-dimensional symmetric tensor. This eleven invariant irreducible functional basis is a proper subset of the Olive-Auffray minimal integrity basis of the tensor.

Even though Eshelby tensors have wide applications in mechanics, its minimal integrity basis and irreducible functional basis haven't been decided yet.
Because the Eshelby tensor has a weaker symmetry than the elasticity tensor, the Eshelby tensor owns more independent elements. Moreover, in three-dimensional physical space, as a conclusion previously introduced, elasticity tensors have a minimal integrity basis with 297 invariants. For these two reasons, it may need a large number of invariants to form a minimal integrity basis of three-dimensional Eshelby tensors. Consequently, in this article we only study invariants of Eshelby tensors in two-dimensional physical space.





Complex variable method is our fundamental tool for recovering Eshelby tensors from a set of isotropic invariants, i.e., integrity basis. It was established by Pierce \cite{Pierce95}, and further applied to plane elasticity by Vianello \cite{complex}. In particular, Olive \cite{olive2017} gave a summary of this method in an algebraic viewpoint. By the orthogonal irreducible decomposition in \cite{Zheng2006}, the Eshelby tensor is factorized into six parts: three scalars $\lambda,\  \mu,\  \upsilon$, two second order irreducible (i.e., symmetric and traceless) tensor ${\mathbf D^1},\ {\mathbf D^2}$, a fourth order irreducible tensor ${\mathbf D}$. Therefore, we only need to study the actions on ${\mathbf D^1},\ {\mathbf D^2},\ {\mathbf D}$ of two-dimensional orthogonal group $\rm O_2$. It is a familiar conclusion in invariant theory that for any fixed positive integer $m$, the dimension of an $m$th order two-dimensional irreducible tensor space is two, which helps to construct a one-to-one correspondence between an irreducible tensor and a complex number. Moreover, there exists an isomorphism between the action of $\rm O_2$ on a second order irreducible tensor and the action of the same group on products of complex planes. Based on these two elementary facts, complexifying the problem becomes a useful approach to study the action of $\rm O_2$ on the second order irreducible tensor space.
The structure of this paper is as follows. To make our statement as self-contained as possible, we first give some notations and briefly review some basic definitions in theory of representations for tensor functions in Section \ref{Pre}. In Section \ref{22}, starting from the irreducible decomposition of $\bf M^{(2)}$, we further review the complex variable method and propose a set of ten polynomial isotropic invariants of $\bf M^{(2)}$. In Section \ref{33}, we prove that these ten invariants are functional irreducible. Consequently, we obtain an irreducible polynomial functional basis of $\bf M^{(2)}$, which is the main goal of this paper.


\section{Preliminaries}\label{Pre}

In this paper, we denote 2D as two-dimensional physical space, $\mathbb{H}^m$ as the $m$th order irreducible tensor space in 2D.
As classical terminology, $\rm O_2$ is the group of orthogonal transformations in 2D, and $\rm SO_2$ is the rotation subgroup of $\rm O_2$. $e_1:=(1,\ 0)^T$, $e_2:=(0,\ 1)^T$ are a pair of familiar orthonormal bases in 2D. ${\widetilde{Q}}$ is the reflection transformation such that ${\widetilde{Q}}e_1=e_1,\ {\widetilde{Q}}e_2=-e_2.$ Obviously, $\rm O_2$ is generated by $\rm O_2$ and ${\widetilde{Q}}$.

Let $\bf T$ be an $m$th order tensor represented by $T_{i_1i_2\ldots i_m}$ under some orthonormal coordinate. A scalar-valued polynomial function $f(T_{i_1i_2\ldots i_m})$ is called an polynomial isotropic invariant of $\bf T$, if the function value is independent of the selection of coordinate system, which means
\begin{equation}\label{O2-act}
f(T_{i_1i_2\ldots i_m})=f(Q_{i_1j_1}Q_{i_2j_2}\ldots Q_{i_mj_m}T_{j_1j_2\ldots j_m}).
\end{equation}
We could rewrite \eqref{O2-act} in a short form
$$f({\bf T})=f(Q\ast {\bf T}),$$
where $Q$ is an arbitrary orthogonal matrix, and $*$ calling the $\rm O_2$-action, defined as the right side in \eqref{O2-act}. Moreover, a set of tensors
$$\rm{O}_2\ast{\bf T} = \{g\ast{\bf T}:g\in \rm{O}_2\}$$
is called the $\rm O_2$-orbit of $\bf T$. Once we have one polynomial invariant, we could easily construct infinite number of polynomial invariants from it. For this reason, our main goal is to find a finite set of polynomial invariants separating the $\rm O_2$-orbits. Therefore, we introduce the definitions of integrity basis and functional basis as below.

\begin{Definition}\label{int}
{\bf (integrity basis)} Let $\{f_1,f_2,\ldots,f_n\}$ be a finite set of polynomial isotropic invariants of  $\bf T$. If any polynomial isotropic invariant of  $\bf T$ is polynomial in $f_1,f_2,\ldots,f_n$, we call the set $\{f_1,f_2,\ldots,f_n\}$ a set of integrity basis of  $\bf T$. In addition, an integrity basis is minimal
if no proper subset of it is an integrity basis.
\end{Definition}

Similarly, we give the definition of functional basis.

\begin{Definition}\label{func}
{\bf (functional basis)} Let $\{f_1,f_2,\ldots,f_n\}$ be a finite set of polynomial isotropic invariants of  $\bf T$. If
$$f_i({\bf T_1})=f_i({\bf T_2}),\ \text{for all}\ i=1,2,...,n$$
implies ${\bf T_1}=g*{\bf T_2}$ for some $g\in{\rm O_2}$, we call the set $\{f_1,f_2,\ldots,f_n\}$ a set of functional basis of  $\bf T$. In addition, a functional basis is minimal (or irreducible)
if no proper subset of it is a functional basis.
\end{Definition}

As known in invariant theory, the algebra of invariant polynomials on finite-dimensional representation $V$ of $\rm O_2$ is finitely generated \cite{Hi-93}. In other words, it claims the existence of a set of integrity basis of $\bf M^{(2)}$. Since integrity bases are also functional bases [3] (but not vise versa), both integrity bases and functional bases could separate orbits.

\section{Orthogonal Transformations and Ten Isotropic Invariants of $\bf M^{(2)}$ }\label{22}

\subsection{Orthogonal Irreducible Decomposition of $\bf M^{(2)}$}

Due to the minor index symmetry of $\bf M^{(2)}$, we have $M_{ijkl}^{(2)}=M_{jikl}^{(2)}=M_{ijlk}^{(2)}$ and the irreducible decomposition of $\bf M^{(2)}$ takes the form \cite{eshelby}:
\begin{equation}\label{de}
M_{ijkl}^{(2)}={\bf{\lambda}} \delta_{ij}\delta_{kl} +2{\bf{\mu}}\delta_{i\hat{k}}\delta_{j\hat{l}} +v (\delta_{i\hat{k}}\epsilon_{j\hat{l}}+\delta_{j\hat{k}}\epsilon_{i\hat{l}})  +\delta_{ij}D_{kl}^{1}+\delta_{kl}D_{ij}^{2}+D_{ijkl},
\end{equation}
where
\begin{equation}\label{sc}
\lambda=\frac{3}{8}M_{iikk}^{(2)}-\frac{1}{4}M_{ikik}^{(2)},\quad \mu=\frac{1}{4}M_{ikik}^{(2)}-\frac{1}{8}M_{iikk}^{(2)}\quad \text{and}\quad v=\frac{1}{4}\epsilon_{ij}M_{ikjk}^{(2)}
\end{equation}
are three scalars,
\begin{equation}\label{se}
D_{ij}^{1}=\frac{1}{2}M_{kkij}^{(2)}-\frac{1}{4}M_{kkll}^{(2)}\delta_{ij}\quad \text{and}\quad D_{ij}^{2}=\frac{1}{2}M_{ijkk}^{(2)}-\frac{1}{4}M_{kkll}^{(2)}\delta_{ij}
\end{equation}
are two second order irreducible tensors and $D_{ijkl}$ is a fourth order irreducible tensor deduced by (\ref{de}), (\ref{sc}) and (\ref{se}). Here, symbol \   $\hat{\cdot}$ \ in subscripts means calculating the average of circulant symmetric monomials. For example, $\delta_{i\hat{k}}\delta_{j\hat{l}}=\frac{1}{2}(\delta_{ik}\delta_{jl}+\delta_{il}\delta_{jk}).$

\subsection{Orthonormal Basis of $\mathbb{H}^2$ and $\mathbb{H}^4$ and Orthogonal Transformations}\label{Orth}
It is known that for any fixed positive integer $m$, the dimension of $\mathbb{H}^m$ is two. Inspired by \cite{complex}, we could find two appropriate pairs of orthonormal bases in $\mathbb{H}^2$ and $\mathbb{H}^4$, respectively. In particular, we choose
$$E_1=\frac{\sqrt{2}}{2}(e_1\otimes e_1-e_2\otimes e_2) \qquad \text{and}\qquad E_2=\frac{\sqrt{2}}{2}(e_1\otimes e_2+e_2\otimes e_1)$$
as a pair of  orthonormal bases in $\mathbb{H}^2$,
\begin{align*}
\mathbb{E}_1=\frac{\sqrt{8}}{8}
(e_1\otimes e_1\otimes e_1\otimes e_1
+e_2\otimes e_2\otimes e_2\otimes e_2
-e_1\otimes e_1\otimes e_2\otimes e_2
-e_1\otimes e_2\otimes e_1\otimes e_2\\
-e_2\otimes e_1\otimes e_1\otimes e_2
-e_2\otimes e_1\otimes e_2\otimes e_1
-e_1\otimes e_2\otimes e_2\otimes e_1
-e_2\otimes e_2\otimes e_1\otimes e_1)
\end{align*}
and
\begin{align*}
\mathbb{E}_2=\frac{\sqrt{8}}{8}
(e_1\otimes e_1\otimes e_1\otimes e_2
+e_1\otimes e_1\otimes e_2\otimes e_1
+e_1\otimes e_2\otimes e_1\otimes e_1
+e_2\otimes e_1\otimes e_1\otimes e_1\\
-e_2\otimes e_2\otimes e_2\otimes e_1
-e_2\otimes e_2\otimes e_1\otimes e_2
-e_2\otimes e_1\otimes e_2\otimes e_2
-e_1\otimes e_2\otimes e_2\otimes e_2)
\end{align*}
as a pair of orthonormal bases in $\mathbb{H}^4$. Here $\otimes$ stands for the tensor product. Similar to the representation of a vector in 2D Cartesian coordinates,
let $\theta_1,\ \theta_2,\ \theta_3$ be the angles between $\bf D^1$ and $E_1$, $\bf D^2$ and $E_1$, $\bf D$ and $\mathbb{E}_1$, respectively. Besides, some definitions are necessary:
\begin{align*}
&H_1:=|{\bf D^1}|\cos(\theta_1)={\bf D^1}\cdot E_1,  \quad H_2:={\bf D^1}|\sin(\theta_1)={\bf D^1}\cdot E_2,\\
& L_1:=|{\bf D^2}|\cos(\theta_2)={\bf D^2}\cdot E_1,  \quad\ L_2:={\bf D^2}|\sin(\theta_2)={\bf D^2}\cdot E_2,
\end{align*}
and
\begin{align*}
&K_1:=|{\bf D}|\cos(\theta_3)={\bf D}\cdot \mathbb{E}_1, \quad K_2:=|{\bf D}|\sin(\theta_3)={\bf D}\cdot \mathbb{E}_2.
\end{align*}

With some calculations, a rotation $Q(\theta)\in {\rm SO_2}$ could satisfy the following identities:
\begin{equation}\label{ro}
\begin{aligned}
&Q(\theta)*E_1=\cos(2\theta)E_1+\sin(2\theta)E_2,\ &Q(\theta)*E_2=-\sin(2\theta)E_1+\cos(2\theta)E_2,\\
&Q(\theta)*\mathbb{E}_1=\cos(4\theta)\mathbb{E}_1+\sin(4\theta)\mathbb{E}_2,\ &Q(\theta)*\mathbb{E}_2=-\sin(4\theta)\mathbb{E}_1+\cos(4\theta)\mathbb{E}_2.
\end{aligned}
\end{equation}
On the other hand, as for $\widetilde{Q}$, we have
\begin{equation}\label{re}
\begin{aligned}
\widetilde{Q}*E_1=E_1,\quad \widetilde{Q}*E_2=-E_2, \quad \widetilde{Q}*\mathbb{E}_1=\mathbb{E}_1,\quad \widetilde{Q}*\mathbb{E}_2=-\mathbb{E}_2.
\end{aligned}
\end{equation}
Resultingly, each $Q(\theta)\in{\rm SO_2}$ acts on $\mathbb{H}^2$ or $\mathbb{H}^4$ as a rotation of $2\theta$ or $4\theta$, respectively. While under the reflection transformation $\widetilde{Q}$, $E_1$ and $\mathbb{E}_1$ are unchanged but $E_2$ and $\mathbb{E}_2$ are turned to the opposite ones. In view of this conclusion,  a one-to-one mapping from an irreducible tensor in $\mathbb{H}^2$ or $\mathbb{H}^4$ to a complex number could be constructed. More precisely, ${\bf D^1},\ {\bf D^2}$ and  ${\bf D}$ are related to complex numbers $z_1$,\ $z_2$ and $z_3$ in sequence, described as
\begin{equation*}\label{z}
\begin{aligned}
&z_1=H_1+H_2\cdot i=|{\bf D^1}|\cdot e^{i\theta_1},\\
&z_2=L_1+L_2\cdot i=|{\bf D^2}|\cdot e^{i\theta_2},\\
&z_3=K_1+K_2\cdot i=|{\bf D}|\cdot e^{i\theta_3},
\end{aligned}
\end{equation*}
where $i$ is the imaginary unit. In other words, we regard the component of ``horizontal" axes as the real part of $z$, and the component of ``vertical" axes as the imaginary part of $z$.
Moreover, the action of $Q(\theta)$ or $\widetilde{Q}$ on the spaces $\mathbb{H}^2$ and $\mathbb{H}^4$ could be seen as an action on the complex plane $\mathbb{C}$. More precisely, according to \eqref{ro} and \eqref{re}, we have
\begin{equation}\label{complex1}
Q(\theta)*(z_1,\ z_2,\ z_3)=(z_1\cdot e^{i(2\theta)},\ z_2\cdot e^{i(2\theta)},\ z_3\cdot e^{i(4\theta)})\in\mathbb{C}^3
\end{equation}
and
\begin{equation}\label{complex2}
\widetilde{Q}*(z_1,\ z_2,\ z_3)=(\bar{z}_1,\ \bar{z}_2,\ \bar{z}_3)\in\mathbb{C}^3\\
\end{equation}
for any $\theta\in [0,\ 2\pi).$
\subsection{Polynomial Invariants of $\bf M^{(2)}$}
Based on orthogonal irreducible decompositions (\ref{de}) and formulas (\ref{complex1}) and (\ref{complex2}), It is now ready to search for a set of polynomial invariants of $\bf M^{(2)}$.
For any polynomial function $p$ of $\bf M^{(2)}$, we can rewrite $p(\bf M^{(2)})$ on the space $\mathbb{R}^3\times \mathbb{C}^3$:
\begin{equation}\label{poly}
p({\bf M^{(2)}})=\Sigma \ C_{abcdefgjk}\lambda^a\mu^b v^c z_1^d \bar{z}_1^e z_2^f \bar{z}_2^g z_3^j \bar{z}_3^k,
\end{equation}
where $a,b,c,d,e,f,g,j$ and $k$ are nine nonnegative integers. $C_{abcdefgjk}\in\mathbb{C}$ is coefficient of each monomial. Since $p({\bf M^{(2)}})$ is a real-valued polynomial, we have
\begin{equation*}\label{eq1}
C_{abcdefgjk}=\bar{C}_{abcedgfkj}.
\end{equation*}
In addition, because $\rm O_2$ is generated by $\rm SO_2$ and ${\widetilde{Q}}$, $p({\bf M^{(2)}})$ should be invariant under any rotation $Q(\theta)$ and the reflection $\widetilde{Q}$, which yields
\begin{equation}\label{O2}
p({\bf M^{(2)}})=p(Q(\theta)*{\bf M^{(2)}})\qquad \text{and}\qquad  p({\bf M^{(2)}})=p(\widetilde{Q}*{\bf M^{(2)}}),
\end{equation}
where $\theta$ is an arbitrary angle. In a viewpoint of complex field, the action of $Q(\theta)$ takes the form
$$p(Q(\theta)*{\bf M^{(2)}})=\Sigma \ C_{abcdefgjk}\lambda^a\mu^b v^c z_1^d \bar{z}_1^e z_2^f \bar{z}_2^g z_3^j \bar{z}_3^k \exp\{i\theta[2(d-e)+2(f-g)+4(j-k)]\}, $$
while the action of $\widetilde{Q}$ takes the form
$$p(\widetilde{Q}*{\bf M^{(2)}})=\Sigma \ C_{abcdefgjk}\lambda^a\mu^b v^c z_1^e \bar{z}_1^d z_2^g \bar{z}_2^f z_3^k \bar{z}_3^j. $$
Combining with \eqref{O2}, we conclude that degrees of each monomial in \eqref{poly} satisfy the Diophantine equation
\begin{equation}\label{eq3}
(d-e)+(f-g)+2(j-k)=0,
\end{equation}
and coefficients are restricted to
\begin{equation}\label{eq2}
C_{abcdefgjk}=C_{abcedgfkj}.
\end{equation}
From \eqref{O2} and \eqref{eq2}, we know that each coefficient $C_{abcdefgjk}$ is a real number. Furthermore, each monomial $C_{abcdefgjk}\lambda^a\mu^b v^c z_1^d \bar{z}_1^e z_2^f \bar{z}_2^g z_3^j \bar{z}_3^k$
should obey the Diophantine equation (\ref{eq3}).

A solution of Diophantine equation is called irreducible if it is not the sum of two or more nonnegative and nontrivial solutions. The following proposition gives a maximal irreducible solution of (\ref{eq3}) and deduce that any nonnegative solution of (\ref{eq3}) is a sum of these irreducible solutions.
For convenience, we denote $w=(d, e, f, g, j, k)$ as a vector of six components.
\begin{Proposition}\label{pr1}
Let
\begin{equation*}
\begin{array}{lll}
  w_1=(1, 1, 0, 0, 0, 0), & w_2=(0, 0, 1, 1, 0, 0), & w_3=(0, 0, 0, 0, 1, 1), \\
  w_4=(2, 0, 0, 0, 0, 1), & w_5=(0, 0, 2, 0, 0, 1), & w_6=(1, 0, 0, 1, 0, 0), \\
  w_7=(1, 0, 1, 0, 0, 1), & w_8=(0, 1, 1, 0, 0, 0), & w_9=(0, 1, 0, 1, 1, 0), \\
  w_{10}=(0,2, 0, 0, 1, 0), & w_{11}=(0,0, 0, 2, 1, 0). &
\end{array}
\end{equation*}
Then, {\rm (i)} $w_1$,\ldots,$w_{11}$ are eleven irreducible solutions of (\ref{eq3}). {\rm (ii)} Each non-negative solution of (\ref{eq3}) is a sum of these irreducible solutions.
\end{Proposition}
\begin{proof}
Property (i) can be easily verified. In order to prove property (ii), we denote
$\Gamma=d+e+f+g+j+k$
and complete the proof by mathematical induction.

When $\Gamma=2,\ 3$, it is easy to testify $w_1,\ldots,w_{11}$ form the whole feasible solutions of (\ref{eq3}) in these two cases. To take a further step, we assume that if the sum of six components of a solution is not more than $\Gamma$ ($\ge 3$), then this solution could be a sum of $w_1,\ldots,w_{11}$. Now we consider a new feasible solution $w=(d, e, f, g, j, k)$ satisfying $d+e+f+g+j+k=\Gamma+1$ and $d-e+f-g+2(j-k)=0$. We finish the proof within two cases.

\textbf{Case1:}\ if $j=k$, then we have $d-e+f-g=0$.

In this case, if we further assume $d=e$, which reduces to $f=g$ and $j+k=\Gamma+1\ge 4$, which implies $j,\ k\ge1$. Therefore, $(d, e, f, g, j-1, k-1)$ with the sum $\Gamma-1$ is also a solution. By the assumption, $(d, e, f, g, j-1, k-1)$ can be represented as the sum of $w_1,\ldots,w_{11}$. Combined with $(d, e, f, g, j, k)=(d, e, f, g, j-1, k-1)+w_3$, the conclusion is valid in this subcase.

If $d\neq e$, without loss of generality, let $d>e$, then $g>f$. Similarly, we have $d,\ g\ge1$ and $(d-1, e, f, g-1, j, k)$ can be represented as the sum of $w_1,\ldots,w_{11}$, indicating that the conclusion is also valid.

\textbf{Case2:}\ if $j\neq k$, we could assume $j>k$, thus $e+g\ge2+d+f\ge2$. Noticing that
\begin{equation*}
\begin{array}{l}
(d, e, f, g, j, k)=(d, e-1, f, g-1, j-1, k)+w_{9}:=u_1+w_{9},\\
(d, e, f, g, j, k)=(d, e-2, f, g, j-1, k)+w_{10}:=u_2+w_{10},\\
(d, e, f, g, j, k)=(d, e, f, g-2, j-1, k)+w_{11}:=u_3+w_{11},
\end{array}
\end{equation*}
and at least one of $u_1,\ u_2,\ u_3$ is a non-negative solution when $e+g\ge2$, which completes the proof.
\end{proof}

Then we relate each solution $(d, e, f, g, j, k)$ to a complex monomial $z_1^d \bar{z}_1^e z_2^f \bar{z}_2^g z_3^j \bar{z}_3^k$, so that eleven solutions $w_1,\ldots,w_{11}$ could correspond to eleven complex monomials. Since invariants are real-valued functions, we only need to consider the real parts of these complex monomials. Hence, we obtain seven different polynomial invariants $J_1,\ldots,J_7$ of $\bf M^{(2)}$ and their relations to ${\bf D^{1}},\ {\bf D^{2}}\ \text{and}\ {\bf D}$ are presented concurrently:
\begin{equation}\label{in}
\begin{aligned}
&w_1\rightarrow J_1:={\rm Re}(z_1\bar{z}_1)=|z_1|^2=H_1^2+H_2^2=D_{ij}^1\cdot D_{ij}^1,\\
&w_2\rightarrow J_2:={\rm Re}(z_2\bar{z}_2)=|z_2|^2=L_1^2+L_2^2=D_{ij}^2\cdot D_{ij}^2,\\
&w_3\rightarrow J_3:={\rm Re}(z_3\bar{z}_3)=|z_3|^2=K_1^2+K_2^2=D_{ijkl}\cdot D_{ijkl},\\
&w_4,\ w_{10}\rightarrow J_4:={\rm Re}(z_1^2\bar{z}_3)=(H_1^2-H_2^2)K_1+2H_1H_2K_2=D_{ij}^1\cdot D_{ijkl}\cdot D_{kl}^1,\\
&w_5,\ w_{11}\rightarrow J_5:={\rm Re}(z_2^2\bar{z}_3)=(L_1^2-L_2^2)K_1+2L_1L_2K_2=D_{ij}^2\cdot D_{ijkl}\cdot D_{kl}^2,\\
&w_6,\ w_8\rightarrow J_6:={\rm Re}(z_1\bar{z}_2)=H_1L_1+H_2L_2=D_{ij}^1\cdot D_{ij}^2,\\
&w_7,\ w_9\rightarrow J_7:={\rm Re}(z_1z_2\bar{z}_3)=H_1K_1L_1+H_1K_2L_2-H_2K_1L_2+H_2K_2L_1
=D_{ij}^1\cdot D_{ijkl}\cdot D_{kl}^2.
\end{aligned}
\end{equation}
In addition, we denote
$$J_8:=\lambda, \ J_9:=\mu, \ J_{10}:=v.$$
As a result of the discussion above, we finally obtain a set of ten polynomial isotropic invariants $\{J_1,\ldots,J_{10}\}$ of $\bf M^{(2)}$. In the next section, we will prove that these ten isotropic invariants are both minimal integrity bases and irreducible function bases of $\bf M^{(2)}$.
\section{Minimal Integrity Bases and Irreducible Functional Bases of $\bf M^{(2)}$}\label{33}

Now our aim is to prove $J_1,\ldots,J_{10}$ are both minimal integrity bases and irreducible function bases of $\bf M^{(2)}$. We first confirm that any isotropic polynomial invariant is polynomial in $J_1,\ldots,J_{10}$, which infers $J_1,\ldots,J_{10}$ are integrity bases. As we have mentioned, an integrity basis is always a functional basis, therefore, $J_1,\ldots,J_{10}$ also form a set of function basis of $\bf M^{(2)}$. Next, we claim that $J_1,\ldots,J_{10}$ are functionally irreducible. Consequently, they are proved to be a set of irreducible functional basis of $\bf M^{(2)}$ , which is the main goal of this paper.

First, we give the following proposition to show that $J_1,\ldots,J_{10}$ form a set of integrity basis of $\bf M^{(2)}$.
\begin{Proposition}\label{represent}
Any isotropic polynomial invariant of $\bf M^{(2)}$ is polynomial in $J_1,\ldots,J_{10}$.
\end{Proposition}
\begin{proof}
In the beginning, we rewrite the forms of $J_1,\ldots,J_7$ in (\ref{in}) by using $H:=|{\bf D^1}|,\ L:=|{\bf D^2}|,\ K:=|{\bf D}|$ as three norms of ${\bf D^1},\ {\bf D^2},\ {\bf D},$ and $\theta_1,\ \theta_2,\ \theta_3$ as three angles defined as in Section \ref{Orth}. More explicitly, we have
\begin{equation}\label{ex22}
\begin{array}{l}
J_1:=H_1^2+H_2^2=H^2$,\quad $J_2:=L_1^2+L_2^2=L^2$, \quad $J_3:=K_1^2+K_2^2=K^2,\\
J_4:=(H_1^2-H_2^2)K_1+2H_1H_2K_2=H^2K\cdot \cos(2\theta_1-\theta_3),\\
J_5:=(L_1^2-L_2^2)K_1+2L_1L_2K_2=L^2K\cdot \cos(2\theta_2-\theta_3),\\
J_6:=H_1L_1+H_2L_2=HL\cdot \cos(\theta_1-\theta_2),\\
J_7:=H_1K_1L_1+H_1K_2L_2-H_2K_1L_2+H_2K_2L_1=HKL\cdot \cos(\theta_1+\theta_2-\theta_3).\\
\end{array}
\end{equation}
$H,\ L,\ K,\ \theta_1,\ \theta_2$ and $\theta_3$ are independent to each other. Moreover, we introduce six scalar-valued functions of $H,\ L,\ K,\ \theta_1,\ \theta_2$ and $\theta_3$ as below:
\begin{equation}\label{ex11}
\begin{array}{ll}
 & J_{11}:=H^2L^2K^2\sin(2\theta_1-\theta_3)\sin(2\theta_2-\theta_3),\\
   & J_{12}:=H^3LK\sin(2\theta_1-\theta_3)\sin(\theta_1-\theta_2), \\
 &J_{13}:=H^3LK^2\sin(2\theta_1-\theta_3)\sin(\theta_1+\theta_2-\theta_3),\\
  &J_{14}:=HL^3K\sin(2\theta_2-\theta_3)\sin(\theta_1-\theta_2), \\
 &J_{15}:=HL^3K^2\sin(2\theta_2-\theta_3)\sin(\theta_1+\theta_2-\theta_3),\\
  &J_{16}:=H^2L^2K\sin(\theta_1-\theta_2)\sin(\theta_1+\theta_2-\theta_3).
\end{array}
\end{equation}
By some calculations, we have
\begin{equation*}
\begin{array}{lll}
J_{11}=J_6^2\cdot J_3-J_7^2,\quad &J_{12}=J_1\cdot J_7-J_4\cdot J_6,\quad &J_{13}=J_1\cdot J_3\cdot J_6-J_4\cdot J_7,\\
J_{14}=J_5\cdot J_6-J_2\cdot J_7,\quad &J_{15}=J_2\cdot J_3\cdot J_6-J_5\cdot J_7,\quad  &J_{16}=\frac{1}{2}[ J_1\cdot J_5-J_2\cdot J_4].
\end{array}
\end{equation*}

Thus, $J_{11},\ldots,J_{16}$ are polynomials in $J_1,\ldots,J_{7}$. In view of this, they are also polynomial invariants of $\bf M^{(2)}$ and we only need to testify that any polynomial invariant of $\bf M^{(2)}$ is polynomial in $J_1,\ldots,J_{16}$.

Recalling that each non-zero monomial
$$C_{abcdefgjk}\lambda^a\mu^b v^c z_1^d \bar{z}_1^e z_2^f \bar{z}_2^g z_3^j \bar{z}_3^k$$
should satisfy $C_{abcdefgjk}=C_{abcedgfkj}\in\mathbb{R}$ and the Diophantine equation (\ref{eq3}). Therefore, the remaining work is to prove that any sum of two conjugated monomials
$$W:=C_{abcdefgjk}	\ \{\ \lambda^a\mu^b v^c z_1^d \bar{z}_1^e z_2^f \bar{z}_2^g z_3^j \bar{z}_3^k+\lambda^a\mu^b v^c z_1^e \bar{z}_1^d z_2^g \bar{z}_2^f z_3^k \bar{z}_3^ j\ \}$$
with the degrees satisfying (\ref{eq3}) is polynomial in $J_1,\ldots,J_{16}$. Omit scalars, we denote
\begin{equation*}\label{exp}
\begin{aligned}
\hat{W}&:=z_1^d \bar{z}_1^e z_2^f \bar{z}_2^g z_3^j \bar{z}_3^k+z_1^e \bar{z}_1^d z_2^g \bar{z}_2^f z_3^k \bar{z}_3^j\\
&=2Re\{z_1^d \bar{z}_1^e z_2^f \bar{z}_2^g z_3^j \bar{z}_3^k\}\\
&=2H^{d+e}L^{f+g}K^{j+k}\cdot \cos[(d-e)\theta_1+(f-g)\theta_2+2(j-k)\theta_3],
\end{aligned}
\end{equation*}
and further define eight angles $\beta_1,\ldots,\beta_8$ as:
\begin{equation*}
\begin{array}{llll}
\beta_1=2\theta_1-\theta_3,\quad &\beta_2=2\theta_2-\theta_3,\quad&\beta_3=2\theta_1-\theta_2,\quad &\beta_4=\theta_1+\theta_2-\theta_3,\\
\beta_5=-\theta_1+\theta_2,\quad &\beta_6=-\theta_1-\theta_2+\theta_3,\quad &\beta_7=-2\theta_1+\theta_3,\quad &\beta_8=-2\theta_2+\theta_3.
\end{array}
\end{equation*}
Similar to the proof of Proposition \ref{pr1}, if $(d-e)+(f-g)+2(j-k)=0,$ the linear combination $(d-e)\theta_1+(f-g)\theta_2+2(j-k)\theta_3$ of $\theta_1,\ \theta_2$ and $\theta_3$, would also be a linear combination of $\beta_1,\ldots,\beta_8$, expressed by
\begin{equation*}\label{cos}
(d-e)\theta_1+(f-g)\theta_2+2(j-k)\theta_3=\alpha_1\beta_1+\cdots+\alpha_8\beta_8,
\end{equation*}
where $\alpha_1,\ldots,\alpha_8$ are all natural numbers.

With some simple calculations, $\cos[\alpha_1\beta_1+\cdots+\alpha_8\beta_8]$ is polynomial in $\cos(\beta_1),\ldots,\cos(\beta_8)$ and $\sin \beta_{i}\cdot \sin \beta_{j}$ ($i, j\in\{1,\ldots,8\}$). Considering the forms of $J_1,\ldots,J_{16}$ in \eqref{ex22} and \eqref{ex11}, we claim that each sum of two conjugated monomials $W$ is polynomial in $J_1,\ldots,J_{16}$.
Thus we finish the proof.
\end{proof}

To take a further step, we need to prove that $J_1,\ldots,J_{10}$ are functionally irreducible (then also polynomially irreducible).
\begin{Proposition}\label{function}
$J_1,\ldots,J_{10}$ are functionally irreducible.
\end{Proposition}
\begin{proof}
Due to the orthogonal irreducible decompositions (\ref{de}), it is clear that three scalars $J_8,\ J_9,\ J_{10}$ are functionally irreducible, hence we only need to consider about $J_1,\ldots,J_7$. Our goal is to change the value of $J_s\ (s=1,\ldots,7)$ while the other six invariants are unchanged.

Case1: When $s=1$, let $L=K=0$, which leads to $J_2=\cdots=J_7=0$. However, $J_1$ will change when $H$ changes, so that $J_1$ can not be a function of the others.

Case2: When $s=2$, let $H=K=0$, which leads to $J_1=J_3=\cdots=J_7=0$. However, $J_2$ will change when $L$ changes, so that $J_2$ can not be a function of the others.

Case3: When $s=3$, let $H=L=0$, which leads to $J_1=J_2=J_4=\cdots=J_7=0$. However, $J_3$ will change when $K$ changes, so that $J_3$ can not be a function of the others.

Case4: When $s=4$, let $L=0$,\ $K$ and $H$ be two fixed and non-zero numbers, which leads to $J_2=J_5=J_6=J_7=0$ and $J_1=H^2$, $J_3=K^2$ are unchanged. However,  $J_4=H^2K\cdot \cos(2\theta_1-\theta_3)$ will change when $2\theta_1-\theta_3$ changes, so that $J_4$ can not be a function of the others.

Case5: When $s=5$, let $H=0$,\ $K$ and $L$ be two fixed and non-zero numbers, which leads to $J_1=J_4=J_6=J_7=0$ and $J_2=L^2$, $J_3=K^2$ are unchanged. However,  $J_5=L^2K\cdot \cos(2\theta_2-\theta_3)$ will change when $2\theta_2-\theta_3$ changes, so that $J_5$ can not be a function of the others.

Case6: When $s=6$, let $K=0$,\ $H$ and $L$ be two fixed and non-zero numbers, which leads to $J_3=J_4=J_5=J_7=0$ and $J_1=H^2$, $J_2=L^2$ are unchanged. However,  $J_6=HL\cdot \cos(\theta_1-\theta_2)$ will change when $\theta_1-\theta_2$ changes, so that $J_6$ can not be a function of the others.

Case7: When $s=7$, let $K$,\ $H$ and $L$ be three fixed and non-zero numbers, which leads to $J_1=H^2$, $J_2=L^2$ and $J_3=K^2$ are unchanged. Now let  $\theta_1=\frac{\pi}{2},\ \theta_2=0,\ \theta_3=\frac{3\pi}{4}$, so we have
\begin{center}
$J_4=\frac{\sqrt{2}}{2}H^2K$,\ $J_5=-\frac{\sqrt{2}}{2}L^2K$,\ $J_6=0$,\ $J_7=\frac{\sqrt{2}}{2}HKL$.
\end{center}
However, when $\theta_1=\frac{3\pi}{2},\ \theta_2=0,$ and $\theta_3=\frac{11\pi}{4}$, we have
\begin{center}
$J_4=\frac{\sqrt{2}}{2}H^2K$,\ $J_5=-\frac{\sqrt{2}}{2}L^2K$,\ $J_6=0$,\ $J_7=-\frac{\sqrt{2}}{2}HKL$.
\end{center}
Only the value of $J_7$ changes, so that $J_7$ can not be a function of the others.

In conclusion, $J_1,\ldots,J_{10}$ are functionally irreducible.
\end{proof}

As a result of Propositions \ref{represent} and \ref{function}, finally, we have the following theorem.
\begin{Theorem}\label{basis}
Define
\begin{equation*}
\begin{array}{l}
J_1:=D_{ij}^1\cdot D_{ij}^1,\quad J_2:=D_{ij}^2\cdot D_{ij}^2,\quad J_3:=D_{ijkl}\cdot D_{ijkl},\quad J_4:=D_{ij}^1\cdot D_{ijkl}\cdot D_{kl}^1,\\
J_5:=D_{ij}^2\cdot D_{ijkl}\cdot D_{kl}^2,\quad J_6:=D_{ij}^1\cdot D_{ij}^2,\quad J_7:=D_{ij}^1\cdot D_{ijkl}\cdot D_{kl}^2,\\
J_8:=\lambda, \quad J_9:=\mu, \quad J_{10}:=v,
\end{array}
\end{equation*}
where $\lambda,\ \mu,\ v,\ {\bf D^1},\ {\bf D^2},\ {\bf D}$ are three scalars, two $2nd$ order irreducible tensors and a $4th$ order irreducible tensor respectively in orthogonal irreducible decomposition (\ref{de}). Then $J_1,\ldots,J_{10}$ are both a set of minimal integrity basis and a set of irreducible functional basis of $\bf M^{(2)}$.
\end{Theorem}

Two examples are further exhibited to confirm the correctness of this method.
\begin{Example}
We list three simple polynomial invariants of $\bf M^{(2)}$ and their polynomial representations by presented basis.
\begin{equation*}
\begin{array}{l}
D_{ij}^1\cdot D_{ijkl}\cdot D_{klpq}\cdot D_{pq}^1=(H_1^2+H_2^2)(K_1^2+K_2^2)=J_1\cdot J_3,\\
D_{ij}^2\cdot D_{ijkl}\cdot D_{klpq}\cdot D_{pq}^2=(L_1^2+L_2^2)(K_1^2+K_2^2)=J_2\cdot J_3,\\
D_{ij}^1\cdot D_{ijkl}\cdot D_{klpq}\cdot D_{pq}^2=\frac{1}{2}(H_1L_1+H_2L_2)(K_1^2+K_2^2)=\frac{1}{2}J_3\cdot J_6.\\
\end{array}
\end{equation*}
\end{Example}
\begin{Example}
According to this method, we can also obtain a minimal integrity basis for 2D elasticity tensors (denoted as $C^{(2)}_{ijkl}$), which has been discussed by Vianello in \cite{complex}.
From the irreducible decomposition of {\rm 2D} elasticity tensors, we know that there are two scalars, one $2nd$ order {\rm 2D} irreducible tensor and one $4th$ order {\rm 2D} irreducible tensor, denoted as $\delta$, $\epsilon$, $E^2$ and $E^4$, respectively. Similarly, we know that any monomial
$$C_{abcdef}\delta^a\epsilon^bz_1^c\bar{z}_1^dz_2^e\bar{z}_2^f$$
should satisfy $C_{abcdef}=C_{abdcfe}\in\mathbb{R}$ and a linear Diophantine equation
\begin{equation}\label{exam}
c-d+2(e-f)=0,
\end{equation}
where
\begin{equation*}\label{co}
\begin{aligned}
&z_1:=H_1+H_2\cdot i=|E^2|\cdot e^{i\theta_1},\\
&z_2:=L_1+L_2\cdot i=|E^4|\cdot e^{i\theta_2}.
\end{aligned}
\end{equation*}
It is easy to find a maximal irreducible solution of  Diophantine equation \eqref{exam}:
\begin{equation*}
\begin{array}{lll}
w_1=(c,\ d,\ e,\ f)=(1,\ 1,\ 0,\ 0),\quad w_2=(c,\ d,\ e,\ f)=(0,\ 0,\ 1,\ 1),\\
w_3=(c,\ d,\ e,\ f)=(2,\ 0,\ 0,\ 1),\quad w_4=(c,\ d,\ e,\ f)=(0,\ 2,\ 1,\ 0).
\end{array}
\end{equation*}
These four solutions can be related to three polynomial invariants:
\begin{equation*}
\begin{aligned}
&w_1\rightarrow J_1:=Re(z_1\bar{z}_1)=|z_1|^2=H_1^2+H_2^2=E_{ij}^2\cdot E_{ij}^2,\\
&w_2\rightarrow J_2:=Re(z_2\bar{z}_2)=|z_2|^2=L_1^2+L_2^2=E_{ij}^4\cdot E_{ij}^4,\\
&w_3,\ w_4\rightarrow J_3:=Re(z_1^2\bar{z}_2)=(H_1^2-H_2^2)L_1+2H_1H_2L_2=E_{ij}^2\cdot E_{ijkl}^4\cdot E_{kl}^2,
\end{aligned}
\end{equation*}
which fit the conclusions in \cite{complex}.

\end{Example}

\section*{Acknowledgments}  The authors would like to express their gratitude to Prof. Wennan Zou, for his encouragement during the course of this work
and for many useful discussions.


\end{document}